\documentclass[10pt,fullpage]{article}
\usepackage{amsmath,amssymb,amsfonts,amsthm,epsfig}
\usepackage{cancel}

\usepackage[usenames,dvipsnames]{xcolor}
\usepackage{bm,xspace}
\usepackage{tcolorbox}
\usepackage{cancel}
\usepackage{fullpage}
\usepackage{pras}
\usepackage{framed}
\usepackage{verbatim}
\usepackage{enumitem}
\usepackage{array}
\usepackage{multirow}
\usepackage{afterpage}
\usepackage{mathrsfs}
\usepackage{pifont} 
\usepackage{chngpage}
\usepackage[normalem]{ulem}
\usepackage{boxedminipage}
\usepackage{caption}
\usepackage{tikz}

\newcommand{\OPT}{\mathsf{OPT}}

\newcommand{\citet}[1]{\cite{#1}}

\def\colorful{1}

\ifnum\colorful=1

\fi
\ifnum\colorful=0

\fi

\newcommand{\cost}{\mathrm{cost}}

\makeatletter
\newtheorem*{rep@theorem}{\rep@title}
\newcommand{\newreptheorem}[2]{
\newenvironment{rep#1}[1]{
 \def\rep@title{#2 \ref{##1}}
 \begin{rep@theorem}\itshape}
 {\end{rep@theorem}}}
\makeatother
\theoremstyle{plain}

\newreptheorem{theorem}{Theorem}

\makeatletter
\newtheorem*{rep@claim}{\rep@title}
\newcommand{\newrepclaim}[2]{
\newenvironment{rep#1}[1]{
 \def\rep@title{#2 \ref{##1}}
 \begin{rep@claim}\itshape}
 {\end{rep@claim}}}
\makeatother
\theoremstyle{plain}

\newrepclaim{claim}{Claim}

\newtheorem{Alg}{Algorithm}
\newcommand{\myalg}[4][0cm]{
\medskip
\small{
\fbox{
\parbox{5.2in}{\vspace{#1}
\begin{Alg}\label{#2}{\textsc{ #3}}
\vspace{.1cm}\\ \emph{ #4}
\end{Alg}
}}
\medskip
}}

\begin{document}

\title{
Distortion in metric matching with ordinal preferences
\vspace{15pt}}

\author{Nima Anari\\ \hspace{0pt}{Stanford University}\\\hspace{0pt}{\texttt{anari@stanford.edu}}
\and Moses Charikar \\ \hspace{0pt}{Stanford University}  \\ \hspace{0pt}{\texttt{moses@cs.stanford.edu}}
\and Prasanna Ramakrishnan\\ \hspace{0pt}{Stanford University}  \\ \hspace{0pt}{\texttt{pras1712@stanford.edu}}
}

\date{\vspace{15pt}\small{\today}}

\maketitle

\begin{abstract} 
Suppose that we have $n$ agents and $n$ items which lie in a shared metric space. We would like to match the agents to items such that the total distance from agents to their matched items is as small as possible. However, instead of having direct access to distances in the metric, we only have each agent's ranking of the items in order of distance. Given this limited information, what is the minimum possible worst-case approximation ratio (known as the \emph{distortion}) that a matching mechanism can guarantee?

Previous work by \citet{CFRF+16} proved that the (deterministic) Serial Dictatorship mechanism has distortion at most $2^n - 1$. We improve this by providing a simple deterministic mechanism that has distortion $O(n^2)$. We also provide the first nontrivial lower bound on this problem, showing that any matching mechanism (deterministic or randomized) must have worst-case distortion $\Omega(\log n)$. 

In addition to these new bounds, we show that a large class of truthful mechanisms derived from Deferred Acceptance all have worst-case distortion at least $2^n - 1$, and we find an intriguing connection between \emph{thin matchings} (analogous to the well-known thin trees conjecture) and the distortion gap between deterministic and randomized mechanisms.

\end{abstract}

 \thispagestyle{empty}
\newpage

\section{Introduction}

Social choice theory aims to understand the ways in which the preferences of individuals over several alternatives can be aggregated and used to make a collective decision for the group. Typically the goal is to design mechanisms that maximize some notion of the total welfare of the participants or equivalently minimize some cost. One cost framework that has received increased attention over the last decade is \emph{metric distortion} (see \citet{AFRSV21} for a detailed survey of the area). In this model, the individuals and the alternatives lie in a shared metric space, which is unknown to the mechanism. The participants indicate their preferences over the alternatives by giving an ordinal ranking of the alternatives in order of distance. Based on this information, the mechanism makes a collective decision with the goal of minimizing the \emph{distortion} -- the worst-case approximation ratio between the cost of the mechanism's choice and the optimal choice with full knowledge of the distances in the metric.

The metric distortion framework has been most well understood for elections, where the individuals represent voters, the alternatives represent candidates, and the goal is to select a single candidate with a small total distance to the voters. A long line of work \cite{ABP15,AP17,MW19,Kem20a,Kem20b,GHS20,CR22,PS21,KK22} pinned down the optimal distortion for deterministic mechanisms at 3, and showed that the optimal distortion for randomized mechanisms is between 2.1126 and $3-o(1)$.

In this paper, we turn our attention to understanding the metric distortion model in a less understood but similarly ubiquitous social choice problem: minimum-cost matching. Now the individuals are $n$ \emph{agents}, and the alternatives are $n$ \emph{items} that we would like to match to the agents. The agents and items lie in a shared metric space, unknown to the mechanism, which only receives each agent's ranking of the items in order of distance. The cost of a matching is the total of the distances between the matched pairs, and once again the \emph{distortion} is the worst-case approximation ratio between the cost of the mechanism's chosen matching and the optimal matching. 

Previous work \cite{CFRF+16} first studied this problem with respect to two well-known mechanisms: Serial Dictatorship (SD) and Random Serial Dictatorship (RSD). In these mechanisms, the agents are matched sequentially (either in a fixed order or a random order) to their favorite unmatched item. These mechanisms are particularly attractive because they are \emph{truthful} -- no agent ever has the incentive to misreport their preferences. \citet{CFRF+16} showed that SD guarantees distortion $2^n - 1$, and provided a matching lower bound. They also showed that RSD guarantees distortion $n$, and proved a lower bound of $\Omega(n^{0.29})$. This work motivates the following question, raised by \citet{AFRSV21}.

\begin{question}[{{\cite[Open problem~2]{AFRSV21}}}]\label{q:main}
What is the distortion of the best ordinal rule for the minimum-cost metric matching problem? What if we also require truthfulness?
\end{question}

The metric distortion framework has been attractive since the expressive power of metric spaces allows it to reflect a host of real-world circumstances. In a variety of matching problems with a notion of distance, algorithms for \Cref{q:main} may be useful. Some examples include assigning students to schools, cars to parking spots, patients to doctors, employees to office buildings, jobs to servers, and so on. Note that in many of these examples, the items can accommodate more than one agent, but this can be reduced to the perfect matching setting by having several copies of an item corresponding to its capacity. It is also worth noting that oftentimes, the metric space is not completely hidden -- for example, school districts often know where students live. However, the demand that the algorithm works for any metric space can provide a stronger guarantee that allows for greater flexibility for the agents. For example, if a family finds it more convenient that their child's school is closer to a parent's place of work rather than their home address, a mechanism that only asks for preferences over schools would be able to accommodate this. In other settings, participants may prefer not to reveal their location for privacy reasons. Beyond practical considerations, the metric distortion model gives us a new lens for understanding and evaluating the efficacy of existing social choice mechanisms, and also provides fertile ground motivating the development of novel mechanisms that might find use beyond this framework.

\subsection{Our contributions and technical overview}

\paragraph*{New distortion bounds.} Our main contributions are new upper and lower bounds for \Cref{q:main}. In \Cref{sec:UBs}, we describe a simple, deterministic matching mechanism that guarantees distortion $O(n^2)$. This is an exponential improvement over SD. The idea of the mechanism is to partition the agents into sets, and for each set designate a representative who is close to all of the other agents in the set. This means that the top choices of the representative are good choices for other agents in the set. If two representatives have a common item among their top choices, then they must actually be close to each other, and so the sets can be merged with one of the two representatives becoming the new representative for the merged set. Otherwise, all of the top several choices of the representatives are disjoint, and we can match the members of a set to the top choices of their representative.

In \Cref{sec:LBs} we prove the first nontrivial lower bounds for \Cref{q:main}, showing that any matching mechanism, regardless of whether or not it is randomized, must have distortion $\Omega(\log n)$ in the worst case. The lower bound is based on a metric space where the agents are the leaves of a binary tree, and the items are internal nodes of the tree. The key is that in every subtree, there is one fewer item than agents, and so any matching must match one agent outside of the subtree. By considering a variety of different underlying metric spaces, it can be arranged so that in the optimal matching, only one agent incurs any cost for being matched outside of their subtree, but the mechanism pays for one agent at each of the $\log n$ levels of the tree.

\paragraph*{The impact of truthfulness on distortion.} Truthfulness is a crucial feature of matching mechanisms in applications where agents report their own preferences. Since truthfulness is an additional constraint on a mechanism, a natural question is whether the optimal distortion for truthful mechanisms is worse than that for non-truthful mechanisms. We begin to explore this question in \Cref{sec:truth} by showing that a large class of deterministic truthful mechanisms have distortion at least $2^n - 1$, which is the same as the lower bound proven by \citet{CFRF+16} for SD. 

Our proof works by identifying a general condition (\Cref{def:serial}) on mechanisms where the SD lower bound instance due to \citet{CFRF+16} still applies. Interestingly, this condition captures a wide class of truthful mechanisms that are derived by artificially creating preferences for the items over the agents and running Deferred Acceptance. Such mechanisms are commonly used, especially in the school choice literature \cite{AS03,Shi22,AAL+22}. 

This result can be viewed either as an indication that the optimal distortion for truthful mechanisms could be substantially worse than that for non-truthful mechanisms or as a partial characterization of the way in which a deterministic truthful mechanism must look different from SD and Deferred Acceptance in order for it to have low distortion.

\paragraph*{Thin matchings and rounding randomized mechanisms.} In \Cref{sec:thin}, we seek to understand whether any randomized mechanism with low distortion can be converted into a low distortion deterministic mechanism in a black box fashion. A priori, one may not expect this to be possible. For instance, Serial Dictatorship has exponentially worse distortion than Random Serial Dictatorship. However, we show that it is possible to do such a conversion, with a relatively small blow up in the distortion.

The general idea is to find a matching among agent-item pairs that the randomized mechanism is likely to match. If such a matching had high distortion, then the randomized mechanism would also have high distortion. We use this to show that any randomized mechanism with distortion $D$ can be rounded to a deterministic mechanism with distortion $Dn^2$. Moreover, we show that the blow-up factor in this rounding scheme is characterized by what we call the \emph{thin matching problem}. This problem asks if given a fractional matching, one can always find a perfect matching such that across any cut, the weight of the perfect matching is at most a small factor larger than the weight of the fractional matching. Thin matchings are analogous to the notion of \emph{thin trees} which played a key role in recent developments improving approximation algorithms for the Asymmetric Traveling Salesman Problem \cite{AGM+17,AKS10,GS11,AG15}. The thin matching problem is a clean combinatorial question and might be of independent interest.

\paragraph*{Bounds for \emph{near}-perfect matchings.}

The distortion of the deterministic mechanism in this work, as well as the best-known randomized mechanism \cite{CFRF+16}, are both polynomially large: $O(n^2)$ and $O(n)$ respectively. In light of this, it is natural to ask if there is a way to relax the problem and achieve much better distortion. Indeed, one way to relax the problem called resource augmentation was studied in the prior work of \citet{CFRF+16}. Here we study another relaxation: instead of outputting a perfect matching, what if our algorithm is allowed to output a \emph{near}-perfect matching, a matching between $1-\epsilon$ fraction of the agents and items?\footnote{Note that the cost of the near-perfect matching is still compared to the optimal perfect matching, otherwise this would not be a relaxation.} In \Cref{sec:near-perfect}, we show that this relaxation allows us to lower the distortion significantly. For any constant $\epsilon>0$, there is a randomized mechanism, a simple truncation of the Random Serial Dictatorship, that achieves a distortion of $O(1)$. We also show that this mechanism can be modified to get a high-probability distortion guarantee instead of an in-expectation one. We show that there is a randomized mechanism that outputs a near-perfect matching $M$, that with high probability has distortion $\poly\log(n)$ simultaneously on all metrics compatible with the input preferences. As a simple corollary, we also get deterministic, albeit exponential time, mechanisms that achieve distortion $\poly\log(n)$. 

For an alternate view of these results, note that if we have a near-perfect matching with cost $c$ and a perfect matching with cost $C$, then we can extend the near-perfect matching to a perfect matching with cost $2c + C$. We simply overlay the two matchings, and match any previously unmatched agent to the unmatched item it is connected two via a path. This means that we can also view these results as giving a black box way to get a granular upgrade to any distortion result in the perfect matching setting. For instance, combining this with the mechanism in \Cref{sec:UBs}, we can obtain a deterministic mechanism where the total cost is within a $O(n^2)$ factor of the optimum cost, and where all but an $\eps$ fraction of the agents are matched at a cost within a  $\poly\log(n)$ factor of the optimum.

\subsection{Related work}

Since it was originally introduced by \cite{HZ79}, the problem of matching under incomplete information has been studied under a variety of contexts and models. \cite{FRFZ14} first studied distortion in ordinal matching, but rather than considering costs induced by a metric space, they considered \emph{unit-sum} and \emph{unit-range} valuations (for each agent, their value of the items sums to 1 or their value for each item is in the interval $[0,1]$ respectively). We note that the switch from costs to valuations means that the objective is maximum-matching rather than minimum-matching. In this setting, they showed that RSD guarantees distortion $O(\sqrt{n})$ and that this is the best possible for randomized mechanisms. For deterministic mechanisms with unit-sum valuations, \citet{ABFRV22b} showed that distortion $\Theta(n^2)$ is the best possible. 

Another interesting variant considers unrestricted valuation functions but allows the mechanism to query a small number of the true valuations of each agent. Without these queries, one cannot achieve bounded distortion, but surprisingly with just a few queries per agent, the problem becomes tractable. \citet{ABFRV22b} showed that with $O(\lambda \log n)$ queries per agent, a mechanism can guarantee distortion $O(n^{1/(\lambda + 1)})$, and \citet{ABFRV22a} showed that even with just two queries per agent, one can get distortion $O(\sqrt{n})$.

The metric assumption has also been applied to maximum matching with ordinal preferences rather than minimum matching. There, the best-known distortion bounds are small constants, even for truthful mechanisms. \citet{AS16p,AS16q} gave a mechanism that guarantees distortion $1.6$ and a truthful mechanism that guarantees distortion $1.76$ (both mechanisms are randomized). 

Finally, \citet{AZ21} studied the min-cost ordinal matching problem in metric spaces where the locations of the agents are unknown, but the locations of the items are known. This is natural in some applications such as facility assignment. They showed that this setting is substantially easier, giving a deterministic mechanism with distortion $3$ which is optimal. Their mechanism works by imagining that each agent is at the same location as their favorite item, and then outputting the optimal solution in this ``projected'' instance where all the distances are known.

\section{Preliminaries}\label{sec:prelims}

\paragraph{Notation and terminology.} Throughout the paper we will use $A = \{a_1, a_2, \dots, a_n\}$ to denote the set of agents, and $B = \{b_1, b_2, \dots, b_n\}$ to denote the set of items. An \emph{ordinal matching problem instance} is given by $n$ \emph{preference lists}, one provided by each agent. Each preference list is a permutation of the items, given in increasing order of distance. Agent $a_i$ \emph{prefers} $b_j$ \emph{over} $b_k$ if $d(a_i, b_j) \leq d(a_i, b_k)$, and $a_i$'s \emph{favorite} item in a set $S\subseteq B$ is $\arg\min_{b \in S} d(a_i, b)$. Similarly, $a_i$'s favorite $t$ items are the first $t$ items on their preference list.

\paragraph{Metric spaces.} Let $d:(A\cup B) \times (A\cup B) \to \mathbb{R}_{\geq 0}$ be the distance function of the metric space. By definition, $d$ satisfies the following three properties.

\begin{enumerate}[label=(\arabic*)]
	\item Identity of indiscernibles: $d(x, y) = 0$ if $x = y$, 

	\item Symmetry: $d(x, y) = d(y, x)$,

	\item Triangle inequality: $d(x, y) \leq d(x, z) + d(z, y)$.

\end{enumerate}
In certain definitions of metric spaces, it is required that different elements have positive distance, but we do not require this. This is justified by allowing different agents or items to ``occupy'' the same point in the metric space, in which case their distance is $0$, but different points do indeed have positive distance. We note that just as easily, we can replace $0$ distances with sufficiently small $\epsilon$'s without changing any results. 

\paragraph{Matchings.} We will use bijective functions $M:A \to B$ to represent matchings. Let 
\[\cost(M) \coloneqq \sum_{i\in [n]} d(a_i, M(a_i))\]
denote the cost of the matching $M$. Let $M_\OPT: A \to B$ denote optimal matching, and let $d_i = d(a_i, M_\OPT(a_i))$ so that $\cost(M_\OPT) = \sum_i d_i$. With slight abuse of notation, we will allow $M$ to also denote a set of matched pairs (often treated as edges in a weighted bipartite graph), so that $(a_i, b_j) \in M$ means that $M(a_i) = b_j$.

Given a problem instance, the \emph{distortion} of a particular matching $M$ is $\frac{\cost(M)}{\cost(M_\OPT)}$. The worst-case distortion of a deterministic matching mechanism is its maximum distortion over all problem instances (parametrized by $n$), and the worst-case distortion for a randomized mechanism is the maximum \emph{expected} distortion ($\frac{\mathbb{E}[\cost(M)]}{\cost(M_\OPT)}$) over all problem instances.

When considering randomized mechanisms, we often find it more convenient to view the output as a \emph{fractional} matching rather than a distribution over matchings. In particular, a randomized mechanism is defined by $n^2$ nonnegative real numbers $p_{i, j}$ for $i, j \in [n]$ (semantically, the probability that $a_i$ is matched to $b_j$) such that for each $j$, $\sum_{i = 1}^n p_{i,j} = 1$, and for each $i$, $\sum_{j = 1}^n p_{i,j} = 1$. In this view, the cost of a matching is 

\begin{equation}
\sum_{i = 1}^n \sum_{j = 1}^n p_{i, j}d(a_i, b_j)
\end{equation}

\section{Distortion upper bound}\label{sec:UBs}

\begin{theorem}
There exists a deterministic matching mechanism that guarantees distortion $O(n^2)$.
\end{theorem}

\begin{proof}

 We will maintain a partition of the agents into sets $S_1, S_2, \dots$, and each set will have a representative $r_1, r_2, \dots$ and so on. We will maintain the invariant that for every $a_i \in S_j$, 
\begin{equation}\label{eq:inv1}
d(a_i, r_j) \leq -d_i + 2^{\lambda_j} \sum_{a_k \in S_j} d_k
\end{equation}
for some $\lambda_j$ associated with each set $S_j$. Note that this implies that for each $a_i \in S_j$, $d(M_\OPT(a_i), r_j) \leq d(M_\OPT(a_i), a_i) + d(a_i, r_j) \leq 2^{\lambda_j} \sum_{a_k \in S_j} d_k$. It follows that there are at least $|S_j|$ distinct items $b$ for which
\begin{equation}\label{eq:inv2}
d(b, r_j) \leq 2^{\lambda_j} \sum_{a_k \in S_j} d_k
\end{equation}
In particular, the above inequality holds when $b$ is any of $r_j$'s favorite $|S_j|$ items.

With this in mind, the matching mechanism is the following. 
\begin{center}
\myalg{def:alg1}{RepMatch}{
Initially, there are $n$ singleton sets $S_1, S_2, \dots, S_n$, with $S_i = \{a_i\}$, $r_i = a_i$, and $\lambda_i = 0$.\\
While there exist two sets $S_i$ and $S_j$ such that some item is both one $r_i$'s favorite $|S_i|$ items and one of $r_j$'s favorite $|S_j|$ items: 
\begin{itemize}
	\item Replace $S_i$ and $S_j$ with $S_i \cup S_j$. 
	\item The representative for $S_i \cup S_j$ is $r_i$ if $\lambda_i \geq \lambda_j$ and $r_j$ otherwise. 
	\item The new $\lambda$ for this set is $\max\{\lambda_i, \lambda_j\}$ if $\lambda_i \neq \lambda_j$, and $\lambda_i + 1$ otherwise.
\end{itemize} 
Arbitrarily match the agents in each $S_i$ to the top $|S_i|$ choices of $r_i$.
}
\end{center}

The key is in showing that the invariant (\ref{eq:inv1}) holds throughout the algorithm by observing how the $\lambda$'s must change as sets merge. Clearly, the invariant holds at the start of the algorithm.

Let's suppose that we merge sets $S_1$ and $S_2$, and without loss of generality, $\lambda_1 \leq \lambda_2$. Since $r_2$ becomes the representative of the merged sets, we would like to bound the distance for any $a \in S_1$ to $r_2$. Suppose that $b$ is an item that is both one of $r_1$'s favorite $|S_1|$ items and one of $r_2$'s favorite $|S_2|$ items. Then by (\ref{eq:inv1}) and (\ref{eq:inv2}), 
\begin{align*}
d(a, r_2) \leq d(a, r_1) + d(r_1, b) + d(b, r_2) &\leq \left( -d_i + 2^{\lambda_1} \sum_{a_k \in S_1} d_k\right) + 2^{\lambda_1} \sum_{a_k \in S_1} d_k + 2^{\lambda_2} \sum_{a_k \in S_2} d_k\\
&\leq -d_i + 2^{\max\{\lambda_1 + 1, \lambda_2\}} \sum_{a_k \in S_1\cup S_2} d_k
\end{align*}
which means that our invariant is maintained, and the new $\lambda$ for $S_1 \cup S_2$ is $\max\{\lambda_1 + 1, \lambda_2\}$, as specified by the algorithm. 

This allows us to show another invariant -- throughout the algorithm, $|S_j| \geq 2^{\lambda_j}$. This follows because when $\lambda_1 \leq \lambda_2$, $ 2^{\lambda_1} + 2^{\lambda_2} \geq 2^{\max\{\lambda_1 + 1, \lambda_2\}}$. 

To conclude, the total cost of matching the elements of $S_j$ to the top $|S_j|$ choices of $r_j$ is at most the total distance from the elements of $S_j$ to $r_j$ plus the total distance from $r_j$ to their top $|S_j|$ items. By (\ref{eq:inv1}) and (\ref{eq:inv2}), both of these terms are at most 
$$|S_j|\cdot 2^{\lambda_j} \sum_{a_k\in S_j} d_k\leq |S_j|^2 \sum_{a_k\in S_j} d_k.$$ 
Therefore, the total cost of the final matching output by the algorithm is at most
$$\sum_{j} 2 |S_j|^2 \sum_{a_k\in S_j} d_k \leq 2\max_j |S_j|^2 \cost(M_\OPT)\leq 2n^2 \cost(M_\OPT) $$
as desired. 
\end{proof}

\section{Distortion lower bound}\label{sec:LBs}

\begin{theorem}
Every (possibly randomized) matching mechanism has worst-case distortion $\Omega(\log n)$.
\end{theorem}

\begin{proof}

The goal will be to construct an ordinal matching instance with $2^k$ agents and items (i.e., specify the preference lists of the agents over the items), along with a collection of metric spaces that are consistent with the specified preferences, such that any matching mechanism must have distortion $\Omega(k)$ on one of the metric spaces in the collection.

Our matching instances are constructed as follows. To simplify the description of the preference lists, we will define a metric space such that the preference lists are easy to decipher by ordering distances in the metric space. The metric space is defined in terms of the distances between vertices on a weighted undirected graph. The graph is structured as a tree with $2^k$ leaves\footnote{We define a leaf to be a vertex with no children and an internal vertex to be one with at least one child, so the root is not a leaf despite having degree $1$.}, where the root has one child, and every other internal vertex has two children. Note that this graph has $2^k$ internal vertices. The items occupy the internal vertices and the agents occupy the leaves. The weight of an edge $(u, v)$ for which $u$ is the parent of $v$ is $2^{k - t}$, where $t$ is the distance from $u$ to the root. See \Cref{fig:def-metric} for a helpful diagram.

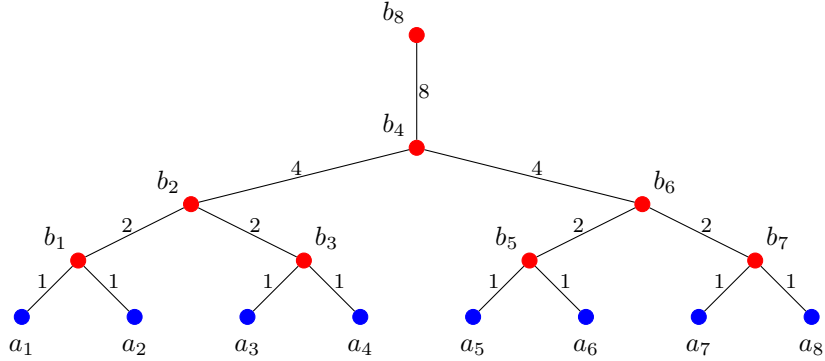
\begin{figure}[h]
\centering
\begin{tikzpicture}[scale=0.75]
\tikzstyle{every node}=[draw,shape=circle, inner sep=2pt];
\node (a1) at (2,0) [fill=blue,color=blue,label=below:$a_1$] {};
\node (a2) at (4,0) [fill=blue,color=blue,label=below:$a_2$] {};
\node (a3) at (6,0) [fill=blue,color=blue,label=below:$a_3$] {};
\node (a4) at (8,0) [fill=blue,color=blue,label=below:$a_4$] {};
\node (a5) at (10,0) [fill=blue,color=blue,label=below:$a_5$] {};
\node (a6) at (12,0) [fill=blue,color=blue,label=below:$a_6$] {};
\node (a7) at (14,0) [fill=blue,color=blue,label=below:$a_7$] {};
\node (a8) at (16,0) [fill=blue,color=blue,label=below:$a_8$] {};
\node (b1) at (3,1) [fill=red,color=red,label=above left:$b_1$] {};
\node (b2) at (5,2) [fill=red,color=red,label=above left:$b_2$] {};
\node (b3) at (7,1) [fill=red,color=red,label=above right:$b_3$] {};
\node (b4) at (9,3) [fill=red,color=red,label=above left:$b_4$] {};
\node (b5) at (11,1) [fill=red,color=red,label=above left:$b_5$] {};
\node (b6) at (13,2) [fill=red,color=red,label=above right:$b_6$] {};
\node (b7) at (15,1) [fill=red,color=red,label=above right:$b_7$] {};
\node (b8) at (9,5) [fill=red,color=red,label=above left:$b_8$] {};
\path[every node/.style={font=\footnotesize}]
(a1) edge node[xshift=-.1cm, yshift=.1cm] {1} (b1) 
(a2) edge node[xshift=.1cm, yshift=.1cm] {1} (b1)
(a3) edge node[xshift=-.1cm, yshift=.1cm] {1} (b3)
(a4) edge node[xshift=.1cm, yshift=.1cm] {1} (b3)
(a5) edge node[xshift=-.1cm, yshift=.1cm] {1} (b5)
(a6) edge node[xshift=.1cm, yshift=.1cm] {1} (b5)
(a7) edge node[xshift=-.1cm, yshift=.1cm] {1} (b7)
(a8) edge node[xshift=.1cm, yshift=.1cm] {1} (b7)
(b1) edge node[xshift=-.1cm, yshift=.1cm] {2} (b2)
(b3) edge node[xshift=.1cm, yshift=.1cm] {2} (b2)
(b5) edge node[xshift=-.1cm, yshift=.1cm] {2} (b6)
(b7) edge node[xshift=.1cm, yshift=.1cm] {2} (b6)
(b2) edge node[xshift=-.1cm, yshift=.1cm] {4} (b4)
(b6) edge node[xshift=.1cm, yshift=.1cm] {4} (b4)
(b4) edge node[xshift=.1cm] {8} (b8);
\end{tikzpicture}
\caption{The metric space that defines the preferences in our ordinal matching instance for $k = 3$. $b_8$ is the root of the tree. The red (internal) nodes are the items and the blue (leaf) nodes are the agents.}\label{fig:def-metric}
\end{figure}

We will repeatedly refer to the structure of this tree as a tool for understanding the structure of the preference lists throughout the proof. The crucial property afforded to us by this construction is the following. 

\begin{claim}\label{claim:crucial-lb}
Consider any vertex $v$ (besides the root) whose parent is $u$. Then the agents in the subtree rooted at $v$ prefer the items in that subtree over $u$ and prefer $u$ over all other items. 
\end{claim}

To see the first part of the claim, note that for any agent $a_i$ and item $b_j$ in the subtree of $v$, we have that 
\[d(a_i, u) = d(a_i, v) + d(v, u) > d(a_i, v) + d(v, b_j) \geq d(a_i, b_j)\]
using the fact that the distance from $v$ to $u$ is greater than the distance from $v$ to any of the vertices in its subtree. The second part follows simply because for any agent $a_i$ in the subtree of $v$, and any $b_j$ not in this subtree, and not $u$, the path from $a_i$ to $b_j$ must go through $u$. 

We now define a set of $2^k$ metric spaces such that any mechanism must have a large distortion with respect to one of these metric spaces. We will have one metric corresponding to each agent $a_i$. Once again, the metrics are defined by distances on a weighted undirected graph. We start with the graph defined above. Then, if $P$ is the set of $k + 1$ items along the path from $a_i$ to the root (including the root), we delete the edges between vertices in $P$ and add edges from $a_i$ to all the items in $P$. The newly added edges have a weight of $1$, and the other edges have a weight of $0$. See \Cref{fig:hard-metric}.

\begin{figure}[h]
\centering
\begin{tikzpicture}[scale=0.75]
\tikzstyle{every node}=[draw,shape=circle, inner sep=2pt];
\node (a1) at (2,0) [fill=blue,color=blue,label=below:$a_1$] {};
\node (a2) at (4,0) [fill=blue,color=blue,label=below:$a_2$] {};
\node (a3) at (6,0) [fill=blue,color=blue,label=below:$a_3$] {};
\node (a4) at (8,0) [fill=blue,color=blue,label=below:$a_4$] {};
\node (a5) at (10,0) [fill=blue,color=blue,label=below:$a_5$] {};
\node (a6) at (12,0) [fill=blue,color=blue,label=below:$a_6$] {};
\node (a7) at (14,0) [fill=blue,color=blue,label=below:$a_7$] {};
\node (a8) at (16,0) [fill=blue,color=blue,label=below:$a_8$] {};
\node (b1) at (3,1) [fill=red,color=red,label=above left:$b_1$] {};
\node (b2) at (5,2) [fill=red,color=red,label=above left:$b_2$] {};
\node (b3) at (7,1) [fill=red,color=red,label=above right:$b_3$] {};
\node (b4) at (9,3) [fill=red,color=red,label=above left:$b_4$] {};
\node (b5) at (11,1) [fill=red,color=red,label=above left:$b_5$] {};
\node (b6) at (13,2) [fill=red,color=red,label=above right:$b_6$] {};
\node (b7) at (15,1) [fill=red,color=red,label=above right:$b_7$] {};
\node (b8) at (9,5) [fill=red,color=red,label=above left:$b_8$] {};
\path[every node/.style={font=\footnotesize}]
(a1) edge node[xshift=-.1cm, yshift=.1cm] {0} (b1) 
(a2) edge node[xshift=.1cm, yshift=.1cm] {0} (b1)
(a3) edge[bend right] node[xshift=-.1cm, yshift=.1cm] {1} (b3)
(a4) edge node[xshift=.1cm, yshift=.1cm] {0} (b3)
(a5) edge node[xshift=-.1cm, yshift=.1cm] {0} (b5)
(a6) edge node[xshift=.1cm, yshift=.1cm] {0} (b5)
(a7) edge node[xshift=-.1cm, yshift=.1cm] {0} (b7)
(a8) edge node[xshift=.1cm, yshift=.1cm] {0} (b7)
(b1) edge node[xshift=-.1cm, yshift=.1cm] {0} (b2)
(a3) edge[bend left] node[xshift=.1cm, yshift=.1cm] {1} (b2)
(b5) edge node[xshift=-.1cm, yshift=.1cm] {0} (b6)
(b7) edge node[xshift=.1cm, yshift=.1cm] {0} (b6)
(a3) edge[bend left] node[xshift=-.1cm, yshift=.1cm] {1} (b4)
(b6) edge node[xshift=.1cm, yshift=.1cm] {0} (b4)
(a3) edge[bend left] node[xshift=-.1cm, yshift=.1cm] {1} (b8);
\end{tikzpicture}
\caption{The consistent metric corresponding to $a_3$.}\label{fig:hard-metric}
\end{figure}
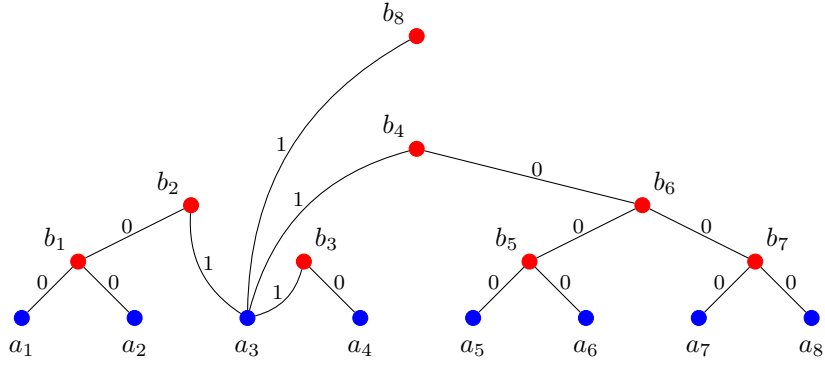

We claim that these metrics are consistent with the preference lists derived before. Agent $a_i$ has a distance of $1$ to every item, so this is consistent with any preference list. For any other agent $a_j$, let $u$ be the first common vertex on the paths from $a_i$ and $a_j$ to the root (in the original tree). Let $v$ be the child of $u$ such that $a_j$ is in the subtree rooted at $v$. Then $a_j$ has distance $0$ to $u$ and all the items in the subtree rooted at $v$, and distance $2$ to all other items. By \Cref{claim:crucial-lb} these distances are consistent with the preference list for $a_j$ that we defined. 

Finally, it remains to show that any mechanism must have a large distortion on one of these metrics. For each of these metrics, $\cost(M_\OPT) = 1$. This is because we can match $a_i$ to the root (which is at distance $1$), and then within every component where the distances are zero, there are an equal number of agents and items. 

We will start by proving the result for deterministic mechanisms and then explain how to extend this to randomized mechanisms. 

Suppose we have a matching $M$ output by a deterministic mechanism. Consider the following process for choosing an agent. Eventually, the metric that causes a large distortion for the mechanism will be the metric corresponding to that agent.

\begin{itemize}

	\item Initially the current vertex is the child of the root.

	\item Until we reach a leaf:

	\begin{itemize}
		
		\item Let the left and right children of the current vertex be $\ell$ and $r$ respectively. 

		\item In the subtree rooted at the current vertex, there is one more agent than item. Therefore, $M$ matches some agent in this subtree to an item outside the subtree. Mark this agent \emph{unlucky}. 

		\item If this agent is in the subtree rooted at $\ell$, then update the current vertex to be $r$, and otherwise update it to be $\ell$.

	\end{itemize}

	\item Choose the agent at the leaf.

\end{itemize}

The process designates one unlucky agent at each step, so there are a total of $k$ unlucky agents. Moreover, in the metric corresponding to the chosen agent, all of the unlucky agents are matched to an item that is distance $2$ away. Hence, these contribute $2k$ to the cost of the mechanism, and the chosen agent also has a cost of $1$, so the total cost is $2k + 1 = \Omega(k)$ as desired.  

Now we will see how to extend this to randomized mechanisms which can now output a fractional matching rather than a perfect matching. We use a similar process to the one described above. The crucial fact will be that within the subtree of the current vertex (any vertex besides the root), the matching gives total weight at least $1$ to pairs $(a_i, b_j)$ where $a_i$ is in the subtree but $b_j$ is not. This is again because the subtree has one more agent than items. 

Instead of saying that this comes from either the left subtree or the right subtree, we can say that at least $1/2$ of this weight must come from either the left or right subtree. If more of the weight comes from the left subtree, the next current vertex is the right child, and otherwise, it is the left child.

All of these pairs $(a_i, b_j)$ have distance $2$, so the total contribution to the cost at each step is at least $1$. Since there are $k$ steps, the total cost is $k + 1 = \Omega(k)$ (where the extra $1$ again comes from the cost of the chosen agent, who has a distance of $1$ to every item). 
\end{proof}

\section{Truthfulness}\label{sec:truth}

Consider the following class of deterministic matching mechanisms.

\begin{definition}\label{def:serial}
Say that a mechanism is \emph{serializable} if there exist permutations $\pi$ and $\sigma$ of $[n]$ with the following property. If the mechanism is given an ordinal matching problem instance where for each $1 \leq i < j \leq n$, $a_{\pi(i)}$ prefers $b_{\sigma(i)}$ over $b_{\sigma(j)}$, then it matches $a_{\pi(i)}$ to $b_{\sigma(i)}$ for all $i$.
\end{definition} 

Intuitively, in these problem instances, the mechanism operates like a Serial Dictatorship where the agents choose in the order $\pi$, and $\sigma$ defines the matching that is outputted. It is easy to see then that every deterministic Serial Dictatorship is serializable. We'll show that all serializable mechanisms fall victim to the same distortion lower bound that applies to Serial Dictatorship.

\begin{claim}
If a mechanism is serializable, then its worst-case distortion is at least $2^n - 1$.
\end{claim}

\begin{proof}
Our proof is essentially a redux of of the lower bound given by \citet{CFRF+16} for Serial Dictatorship. 

Suppose that the underlying metric space is the real number line. For each $i$, the agent $a_{\pi(i)}$ sits at $2^{i - 1}$. For $i < n$, the item $b_{\sigma(i)}$ sits at $2^{i}$, and $b_{\sigma(n)}$ sits at $-\eps$. 

The optimal matching would match $a_{\pi(1)}$ to $b_{\sigma(n)}$, and $a_{\pi(i)}$ to $b_{\sigma(i - 1)}$ for $i > 1$, for a cost of $1 + \eps$. However, \Cref{def:serial} tells us that for all $i$, $a_{\pi(i)}$ is matched to $b_{\sigma(i)}$. This match costs $2^{i - 1}$ for $1 \leq i < n$, and $2^{n - 1} + \eps$ for $i = n$. Thus, the total cost is $2^n - 1 + \eps$. As $\eps$ tends to $0$, the distortion tends to $2^n - 1$ (or equivalently, we can set $\eps = 0$, and break ties in the preferences appropriately).
\end{proof}

Interestingly, the class of serializable mechanisms captures a wide variety of truthful mechanisms derived from Deferred Acceptance \cite{GS62}. 

Recall that Deferred Acceptance is a two-sided matching mechanism where participants on each side of the matching provide ranked preferences of the other side. Imagining that the items also provide a ranking over the agents, the algorithm works as follows. In rounds, each agent without a current match ``proposes'' to their favorite item that they have not yet proposed to. Then each item compares its proposals to its current match (if it has one), keeps its favorite, and rejects the rest. This continues until all of the agents are matched.

A natural way to convert this two-sided mechanism into a one-sided mechanism is to artificially introduce some preferences for the items. This can be done arbitrarily, or with some structure intended to encourage certain outcomes. For example, in the school choice literature \cite{AS03,Shi22,AAL+22}, this is used to prioritize students within walking distance from schools or students from historically disadvantaged neighborhoods. Moreover, since Deferred Acceptance is truthful from the side of the proposers, any one-sided mechanism derived in this way is also a truthful mechanism.  This captures a wide variety of truthful mechanisms, by the sheer number of ways to fix each item's preferences.

However, no such mechanism can have good distortion, by the following. 

\begin{theorem}
If a mechanism is implemented by Deferred Acceptance with fixed preferences for the items, it is serializable.
\end{theorem}

\begin{proof}
Suppose we run a Serial Dictatorship on the \emph{items} (in the order $b_1, b_2, \dots, b_n$), using their preferences to match them to agents. If $b_i$ is matched to $a_j$, let $\pi(i) = j$. If we make $\sigma$ the identity permutation, we claim that $\pi$ and $\sigma$ satisfy the property in \Cref{def:serial}. 

Suppose that per \Cref{def:serial}, we are given an ordinal matching problem instance, where for each $1 \leq i < j \leq n$, $a_{\pi(i)}$ prefers $b_{i}$ over $b_{j}$. We will show by induction on $k$ that $a_{\pi(k)}$ gets matched to $b_k$.

For the base case, note that $a_{\pi(1)}$ and $b_1$ are each other's first choice, so $a_{\pi(1)}$ will propose to $b_1$ in the first round, and this match will be unbreakable. 

Now suppose that we know that $a_{\pi(i)}$ will be matched to $b_i$ for $i \leq k$. Now consider $a_{\pi(k + 1)}$. By the assumption in \Cref{def:serial},  $a_{\pi(k + 1)}$ prefers $b_{k + 1}$ over all of $b_{k + 2}, \dots, b_n$. This means that the only items that $a_{\pi(k + 1)}$ may propose to before $b_{k + 1}$ are $b_1, b_2, \dots, b_k$. But we know that these will be matched to $a_{\pi(1)}, a_{\pi(2)}, \dots, a_{\pi(k)}$ by the inductive hypothesis, and so any matches between $a_{\pi(k + 1)}$ and $b_i$ for $i \leq k$ will eventually be broken. 

Thus, $a_{\pi(k + 1)}$ eventually proposes to $b_{k + 1}$. Now, because of the Serial Dictatorship that was used to define $\pi$, we know that the only candidates that $b_{k + 1}$ might reject $a_{\pi(k + 1)}$ for are $a_{\pi(1)}, a_{\pi(2)}, \dots, a_{\pi(k)}$ (those that were matched in the Serial Dictatorship before $b_{k + 1}$'s turn). But each for each $i \leq k$, $a_{\pi(i)}$ prefers $b_i$ over $b_{k + 1}$, so $a_{\pi(i)}$ would have to propose to $b_i$ before $b_{k + 1}$. Since we know that $a_{\pi(i)}$ is matched to $b_i$ by the inductive hypothesis, $a_{\pi(i)}$ never proposes to $b_{k + 1}$. This means that when $a_{\pi(k + 1)}$ proposes to $b_{k + 1}$, the match is made and never changes, which completes the induction.

Thus, our chosen $\pi$ and $\sigma$ satisfy the requirements of \Cref{def:serial}, as desired.
\end{proof}

\section{Rounding randomized mechanisms and thin matchings}\label{sec:thin}

Suppose that we have a randomized matching mechanism that guarantees some distortion $D$. Can we always round the fractional matching that it outputs to a perfect matching in a way that the cost of the perfect matching is not much larger than the cost of the fractional matching? In particular, we are concerned with the following question.
\begin{question}\label{q:rtod}
What is the smallest $\alpha = \alpha(n)$ such that the following holds? If a fractional matching gives weight $p_{i,j}$ to each pair $(a_i,b_j)\in A\times B$, then there exists a perfect matching $M$ such that for any metric $d$ we have
$$\sum_{(a_i,b_j) \in M} d(a_i,b_j) \leq \alpha \sum_{1\leq i,j\leq n} p_{i,j} d(a_i,b_j).$$
\end{question}

If a randomized mechanism guarantees distortion $D$, then this means that on any input, the fractional matching it outputs satisfies $ \sum_{1\leq i,j\leq n} p_{i,j} d(a_i,b_j) \leq D\cdot \cost(M_\OPT)$. This means that if the above statement is true for some $\alpha$, then any randomized mechanism that guarantees distortion $D$ can be rounded to a deterministic mechanism with distortion $\alpha D$.

Note that we can get $\alpha \leq n^2$ quite easily. By Hall's theorem, the subgraph that consists of the edges $(a_i,b_j)$ with $p_{i,j} \geq \frac{1}{n^2}$ must have a perfect matching\footnote{For any $S\subseteq A$, if $|N(S)| < |S|$, then $\sum_{a_i \in S} \sum_{b_j \notin N(S)} p_{i,j} \geq 1$. But this means for some $(a_i, b_j) \in S \times (B\setminus N(S))$ has $p_{i,j} \geq \frac{1}{n^2}$ which is a contradiction.}, and any such matching $M$ will have 
$$ \sum_{1\leq i,j\leq n} p_{i,j} d(a_i,b_j)\geq \sum_{(a_i,b_j) \in M} p_{i,j} d(a_i,b_j) \geq \frac1{n^2}\sum_{(a_i,b_j) \in M} d(a_i,b_j).$$
In particular, since RSD guarantees distortion $O(n)$  \cite{CFRF+16}, by computing the matching probabilities of RSD and taking a matching from the edges that are matched with probability at least $\frac{1}{n^2}$, we can get a deterministic mechanism that gets distortion $O(n^3)$. One interesting feature of this mechanism is that it entirely offloads the problem of dealing with the preference lists to RSD.

It remains to see what the best factor of $\alpha$ we can guarantee is. We will show that up to an $O(\log n)$ factor, the optimal $\alpha$ is actually given by the following combinatorial question.

\begin{question}[Thin matchings]\label{q:thin-mts}
What is the smallest $\beta = \beta(n)$ such that the following holds? Suppose we have a fractional matching of a balanced bipartite graph $G = (U\cup V, E)$ that gives weight $p_{u,v}$ to each edge $(u, v)$. Then there exists a perfect matching $M$ such that for any cut $(S, \overline{S})$, the weight of $M$ across the cut is within a factor of $\beta$ of the weight of the fractional matching across the cut. i.e., 
\[|M \cap (S\times \overline{S})| \leq \beta\cdot \sum_{(u,v)\in S\times \overline{S}} p_{u,v}.\]
We say that a matching $M$ satisfying the above inequality for all cuts is \emph{$\beta$-thin}.
\end{question}

\begin{proposition}\label{prop:log-factors}
Let $\alpha$ and $\beta$ be the best answers to \Cref{q:rtod} and \Cref{q:thin-mts} respectively. Then 
\[\beta \leq \alpha \leq O(\log n) \cdot \beta.\]
\end{proposition}

\begin{proof}
To prove that $\beta\leq \alpha$, we need to show that if there exists a matching for \Cref{q:rtod} that gets a factor of $\alpha$, then there exists a matching for \Cref{q:thin-mts} that gets a factor of $\alpha$ as well. Indeed, we can use the same matching, and just apply the guarantee for  \Cref{q:rtod} to each of the cut metrics. These metrics assign distance $0$ to edges within a cut, and $1$ to those across the cut, i.e., 
\[ d_S(u,v) := \textbf{1}[(u, v) \in S\times \overline{S}].\]

Next, we prove that  $\alpha \leq O(\log n) \cdot \beta$. We use a well-known consequence of Bourgain's embedding \cite{Bou85}, which says that any metric can be embedded as a conic combination of cut metrics with distortion $O(\log n)$. In particular, for any metric $d$, there exist non-negative weights $\gamma_S$ such that 
\[\sum_{S}\gamma_S d_S(u,v)\leq d(u, v) \leq O(\log n) \cdot \sum_{S}\gamma_S d_S(u,v).\]

This means that if we can get a factor of $\beta$ for \Cref{q:thin-mts}, then we have for any $d$,
\begin{align*}
\sum_{(u,v) \in M} d(u,v) &\leq O(\log n) \sum_{(u,v) \in M} \sum_{S}\gamma_S d_S(u, v) \\
&= O(\log n) \sum_{S} \gamma_S|M \cap (S\times \overline{S})| \\
&\leq O(\log n)\cdot\beta  \sum_S \gamma_S\sum_{(u,v)\in S\times \overline{S}} p_{u,v}\\
&= O(\log n)\cdot \beta \sum_S \gamma_S \sum_{u,v}p_{u,v} d_S(u, v)\\
&= O(\log n)\cdot \beta \sum_{u,v}p_{u,v} \sum_S \gamma_S  d_S(u, v)\\
&\leq O(\log n)\cdot \beta \sum_{u,v}p_{u,v} d(u, v),
\end{align*}
which means exactly that we can get a factor of $O(\log n) \cdot \beta$ for \Cref{q:rtod}.
\end{proof}

In summary, we see if we can prove strong bounds on \Cref{q:thin-mts}, then we can round any randomized mechanism to a deterministic one without losing much in terms of distortion.

\subsection{Relationship with thin trees}

The definition of \emph{thinness} in \Cref{q:thin-mts} closely mirrors the well-known notion of \emph{thin trees} which concerns finding spanning trees, rather than perfect matchings, with few edges across each cut. A definition in the style of \Cref{q:thin-mts} is given below. 

\begin{definition}[Thin trees]
Suppose we have a weighted graph $G$ with edge weights $p_{u, v}$. We say that a spanning tree $T$ is $\beta$-thin if for any cut $(S, \overline{S})$ of the graph,
\[|T \cap (S\times \overline{S})| \leq \beta\cdot \sum_{(u,v)\in S\times \overline{S}} p_{u,v}.\]
\end{definition}
With the above definition, the thin tree conjecture \cite{God04} states that if the edge weights $p_{u, v}$ form a point in the spanning tree polytope,\footnote{That is, the edge weights are a convex combination of indicators of spanning trees, analogous to a fractional matching, but for spanning trees.} then there exists an $O(1)$-thin tree $T$. We remark that in the literature this conjecture is sometimes stated with the alternative assumption that the graph $G$ is $1$-edge-connected, i.e., the weights $p_{u, v}$ sum to at least $1$ across each cut. Up to constant factors in thinness, these assumptions are equivalent due to a classic result of \citet{Nas61}.

Thin trees originated in the study of graph embeddings \cite{God04}, where they were proposed as a path towards Jaeger's flow conjecture \cite{Jae84,Jae88}. They received significant attention much later on when \citet{AGM+17} showed that an efficient algorithm for finding $\beta$-thin trees implies an $O(\beta)$ approximation algorithm for the Asymmetric Traveling Salesman Problem (ATSP). They also gave such an algorithm for $\beta = O(\log n/\log\log n)$. \citet{AG15} gave a proof that $O(\poly\log \log n)$-thin trees always exist, though their proof is non-constructive. The thin trees conjecture postulates that $O(1)$-thin trees always exist (see for example \cite[Conjecture~1.3]{AG15} and the discussion that follows).  See \citet{AKS10,GS11} for some other applications of thin trees in the ATSP literature. 

Just as thin trees have found surprising relevance in very different areas, we find it similarly intriguing that thin matchings are connected to the problem of ordinal metric matching, which on the surface seems unrelated. Note that in \Cref{q:thin-mts}, the notions of preference lists and metric spaces are completely abstracted away, even though they are central to the original problem .

In light of the strong known bounds for thin trees, and the thin trees conjecture, we make the following analogous conjecture.

\begin{conjecture}\label{conj:thin-matchings}
The optimal $\beta$ for \Cref{q:thin-mts} is $O(1)$.
\end{conjecture}

Since \Cref{prop:log-factors} already has a logarithmic factor, for our rounding application an even weaker version of the above conjecture that just shows a $O(\poly \log n)$ bound would suffice -- it implies that the optimal randomized and deterministic mechanisms for ordinal metric matching have the same distortion up to $\poly\log$ factors.

A natural approach to proving strong bounds for \Cref{q:thin-mts} would be to adapt the tools used to make progress on the thin tree conjecture. In particular, the technique used by \citet{AGM+17} was to write the weights $p_{u,v}$ as the mean of a distribution $\mu$ over spanning trees -- this is possible by the assumption that $p_{u,v}$ is in the spanning tree polytope. Amongst all possible $\mu$ with this mean, \citet{AGM+17} pick the one with the maximum entropy. This ensures that for a random sample $T\sim \mu$, each cut value $|T\cap (S\times \overline{S})|$ is highly concentrated around its mean, and by union bounding tail deviation bounds over all the cuts, $O(\log n/\log \log n)$-thinness follows.

An analogous attempt would be to similarly write a fractional matching as the mean of a distribution over perfect matchings that exhibits concentration properties. However, the challenge with matchings is that distributions over them, even max-entropy ones, exhibit strong correlations that prevent useful concentration inequalities. We demonstrate this with an example depicted in \Cref{fig:non-concentration}. Consider a graph $G$ which is a cycle of length $4k$ with vertices numbered $1,2,\dots,4k$ in the order they appear on the cycle. Let $p_{u,v}=1/2$ on every edge of the cycle -- a fractional matching. There are exactly two perfect matchings in $G$, $M_{\text{odd}}$ formed by the odd edges, and $M_{\text{even}}$ formed by the even edges. There is a unique distribution $\mu$ over perfect matchings with mean $p$: $\mu$ has to put $1/2$ mass on each one of $M_{\text{odd}}$ and $M_{\text{even}}$. Now, consider the cut $(S, \overline{S})$ where $S=\{1, 2, 5, 6, \dots, 4i+1, 4i+2, \dots \}$. This cut has all the edges $(2i, 2i+1)$ and none of the edges $(2i+1, 2i+2)$. For a random sample $M\sim \mu$, the value of the cut $|M\cap (S\times \overline{S})|$ is uniformly distributed over the set $\{0, 2k\}$, which shows that no meaningful concentration around the mean holds here.

\begin{figure}[h]
\centering
\begin{tikzpicture}
\tikzstyle{every node}=[draw, shape=circle, fill=black, inner sep=2pt];
\foreach \i in {0, 1, 2, 3, 4, 5, 6, 7, 8, 9, 10, 11}
	\node (a\i) at (30*\i:3) {};
\foreach \u/\v in {0/1, 2/3, 4/5, 6/7, 8/9, 10/11}
	\draw[blue] (a\u) -- (a\v);
\foreach \u/\v in {1/2, 3/4, 5/6, 7/8, 9/10, 11/0}
	\draw[red] (a\u) -- (a\v);
\draw[dashed] plot [smooth cycle] coordinates {(0:2.5) (30:2.5) (60:3.5) (90:3.5) (120:2.5) (150:2.5) (180:3.5) (210:3.5) (240:2.5) (270:2.5) (300:3.5) (330:3.5)};

\end{tikzpicture}
\caption{The uniform distribution over the blue and red perfect matchings shows no meaningful concentration for the number of edges of a sample that cross the indicated cut.}\label{fig:non-concentration}
\end{figure}
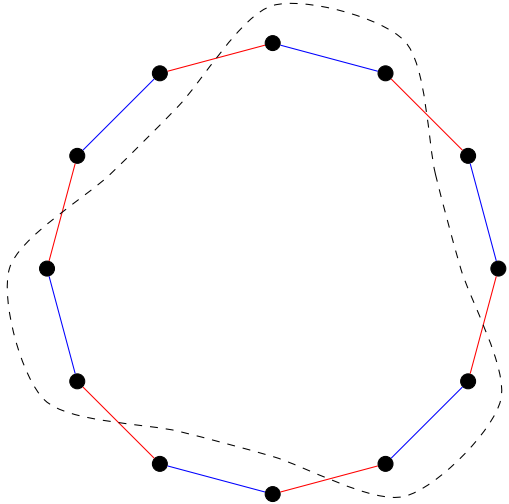

While this example rules out strong concentration around the mean, both $M_{\text{odd}}$ and $M_{\text{even}}$ still produce cut values within a factor of $2$ of the mean, which tempts one to guess that a perhaps weaker multiplicative tail bound can be used to prove \Cref{conj:thin-matchings}. However, we can modify this example to make it even worse. Let $q\in [0, 1]$, and let $p_{u, v}$ be the convex combination of $M_{\text{odd}}$ and $M_{\text{even}}$ with coefficients $q$ and $1-q$ respectively. Then with probability $q$, the sampled perfect matching will have a cut value $1/q$ times as large as the fractional matching. We can boost this failure probability $q$ by parallel repetition. Consider a disjoint union of $n/4k$ cycles of length $4k$ each, where on each cycle the edge weights alternate between $q$ and $1-q$. A max-entropy distribution would independently choose the matching in each cycle. Thus with probability at least $1-(1-q)^{n/4k}$ we will get a cut where the sampled perfect matching has $1/q$ times as many edges as the fractional matching. We can set $k=\Theta(1)$ and $q=\Theta(1/n)$ to get a blow-up by a factor of $\Theta(n)$ in some cut with overwhelming probability.

\section{Near-perfect matchings}\label{sec:near-perfect}

In this section, we study a relaxation of the problem, where our mechanism is allowed to output a non-perfect matching $M$, matching a subset of agents to a subset of items. We still measure the performance of the mechanism in terms of distortion ($\frac{\cost(M)}{\cost(M_\OPT)}$) or expected distortion ($\frac{\mathbb{E}[\cost(M)]}{\cost(M_\OPT)}$), where $M_\OPT$ is the optimal \emph{perfect} matching.

Our first result shows that a simple truncation of the Random Serial Dictatorship (RSD) rule achieves $O(1)$ distortion while matching a $1-\epsilon$ fraction of the agents to items. Random Serial Dictatorship was studied by \citet{CFRF+16} for the non-relaxed problem and was shown to achieve expected distortion $\leq n$. Our analysis of the truncated version follows closely the ideas in their analysis. 

\begin{center}
\myalg{def:alg-truncated-rsd}{TruncatedRSD}{
The goal is to match $m$ out of $n$ agents to items.\\
We start with an empty matching. While fewer than $m$ agents are matched:
\begin{itemize}
	\item Pick an agent $a_i$ uniformly at random from all unmatched agents.
	\item Match $a_i$ to their most preferred unmatched item $b_j$.
\end{itemize}
}
\end{center}

\begin{theorem}\label{thm:rsd-near-perfect}
	If $M$ is the matching output by TruncatedRSD, then
	\[ \mathbb{E}[\cost(M)]\leq \frac{m}{n+1-m}\cdot \cost(M_\OPT). \]
\end{theorem}

Note that by setting $m=(1-\epsilon)n$, we get expected distortion $\leq (1-\epsilon)/\epsilon=O_\epsilon(1)$, yielding the following simple corollary.
\begin{corollary}
	For every constant $\epsilon>0$, there is a randomized mechanism that outputs a matching $M$ of size $(1-\epsilon)n$ with expected distortion $O(1)$.
\end{corollary}
\begin{proof}[Proof of \Cref{thm:rsd-near-perfect}]
	Let $M_i$ be the partial matching after the first $i$ agents have been matched. We will show by induction that there exists a (random) matching $N_i$ between the $n-i$ unmatched agents and the $n-i$ unmatched items ($M_i\cup N_i$ is a perfect matching) such that $\mathbb{E}[\cost(N_i)]\leq \frac{n+1}{n+1-i}\cdot \cost(M_\OPT)$ and $\mathbb{E}[\cost(M_i)]\leq \frac{i}{n+1-i}\cdot \cost(M_\OPT)$. At the end, by setting $i=m$, we get that $\cost(M_m)\leq \frac{m}{n+1-m}\cost(M_\OPT)$ as desired.
	
	For the base case $i=0$, we can take $N_0$ to be $M_\OPT$ and note that $M_0$ is the empty matching. Clearly $\cost(N_0)=\frac{n+1}{n+1}\cost(M_\OPT)$ and $\cost(M_0)=0$.
	
	Now assume that we have proved the claim for $i$. We will prove it for $i+1$. In iteration $i+1$, the agent we pick, let us call it $a$, is uniformly randomly picked from all unmatched ones, i.e., those agents matched in $N_i$. Now $a$ gets matched to their top available item; the cost of this edge is at most the cost of edge incident to $a$ in $N_i$. Therefore the average cost of this edge, averaged over all unmatched $a$, is upper-bounded by $\cost(N_i)/(n-i)$. Therefore
	\begin{align*} \mathbb{E}[\cost(M_{i+1})] &\leq \mathbb{E}[\cost(M_i)]+\frac{\mathbb{E}[\cost(N_i)]}{n-i}\leq \left(\frac{i}{n+1-i}+\frac{n+1}{(n+1-i)(n-i)}\right)\cost(M_\OPT)\\
		&=\frac{i+1}{n-i}\cdot\cost(M_\OPT).
	\end{align*}
	It remains to construct $N_{i+1}$ and bound its cost. Suppose that agent $a$ gets matched to the item $b$ by the mechanism. Since both $a$ and $b$ are unmatched at the beginning of the iteration, they must both be endpoints of edges in $N_i$. Let the other endpoints of these edges be $b'$ and $a'$ respectively. Now we simply let $N_{i+1}$ be the matching obtained from $N_i$ by deleting the edges $(a, b')$ and $(a', b)$, and instead adding the edge $(a', b')$; if $a=a'$ and $b=b'$, we simply delete the edge. It is easy to see that $N_{i+1}$ is a matching with one fewer edge that matches all agents and items still available after this iteration. It remains to bound its cost. Here we use the triangle inequality: $d(a', b')\leq d(a, b')+d(a, b)+d(a', b)$ to get that $\cost(N_{i+1})\leq \cost(N_i)+d(a,b)$; the inequality is still valid in the special case where $a=a'$ and $b=b'$. We already established that $\mathbb{E}[d(a,b)]\leq \mathbb{E}[\cost(N_i)]/(n-i)$. Therefore
	\[ \mathbb{E}[\cost(N_{i+1})]\leq \left(1+\frac{1}{n-i}\right)\cost(N_i)\leq \frac{n+1}{n-i}\cdot \cost(M_\OPT), \]
	as desired.
\end{proof}

Next, we show that by losing some logarithmic factors, we can get a mechanism whose output near-perfect matching $M$ satisfies a much stronger guarantee: with high probability, it has \emph{distortion} (not expected distortion) at most $O(\log^2(n))$ simultaneously w.r.t.\ all metrics $d$ compatible with the input preference rankings. Note that this is not a trivial implication of \Cref{thm:rsd-near-perfect}. A randomized mechanism is allowed to hedge its bets, i.e., choose a mixed strategy against an adversary who chooses the worst metric $d$. A priori, the gap between pure strategies and mixed strategies can be very large.

\begin{theorem}\label{thm:high-prob}
	For any constant $\epsilon>0$, there exists a computationally efficient randomized mechanism that outputs a matching $M$ of size $(1-\epsilon)n$ such that with probability $\geq 1-1/\poly(n)$, $M$ has distortion at most $O(\log^2(n))$ w.r.t.\ the optimal perfect matching in all metrics $d$ compatible with the input preferences.
\end{theorem}
\begin{proof}
	If we run TruncatedRSD with $m=(1-\epsilon/10)n$, at the end we get a random matching of size $m$ with expected distortion $O(1/\epsilon)$. Let $x$ be the fractional (non-perfect) matching which is the expectation of the output of TruncatedRSD. In other words, $x$ is a vector indexed by pairs of agents and items, such that $x_{a,b}=\mathbb{P}[a\text{ matched to }b]$. We know that
	\[ \cost(x)=\sum_{(a, b)\in A\times B}d(a, b)x_{a, b}=\mathbb{E}[\cost(\text{output of TruncatedRSD})]\leq O(1/\epsilon)\cdot \cost(M_\OPT). \]
	Our goal is to round $x$ to an integral matching, without blowing up the cost by much. We employ a strategy based on thin matchings, similar to the proof of \Cref{prop:log-factors}. Before we explain the rounding, we have to deal with a technical issue. We do not know how to efficiently compute the exact vector $x$; it is the expectation of the output of TruncatedRSD, but naively we have to enumerate $\simeq n!$ permutations describing the order in which the agents are picked, in order to compute $x$. Instead, we run TruncatedRSD $h$ independent times, and we let $x'$ be the average of the indicators of the $h$ output matchings. Instead of working with $x$, we work with $x'$. As $h\to \infty$, almost surely $x'\to x$. We will eventually set $h$ to be only polynomially large, making the mechanism computationally efficient.
	
	As in the proof of \Cref{prop:log-factors}, we use Bourgain's embedding \cite{Bou85} to write, up to a distortion of $O(\log(n))$, a hypothetical metric $d$ as a conic combination of cut metrics $d_S$. It is enough to find a (non-perfect) matching $M$ that guarantees for every subset $S$ of our metric space:
	\[ \frac{M(S,\overline{S})}{x(S,\overline{S})}\coloneqq \frac{|M\cap (S\times \overline{S})|}{\sum_{(u, v)\in S\times \overline{S}} x_{u, v}}\leq \beta. \]
	An argument very similar to the proof of \Cref{prop:log-factors} would then show that the distortion of $M$ w.r.t.\ any metric compatible with the input rankings is $\leq O(\beta \log n/\epsilon)$.
	
	We construct this $M$ by employing a randomized rounding mechanism due to \citet{CVZ11}. They showed that for any point $y$ in the matching polytope and a desired parameter $\gamma\in (0, 1/2]$, one can efficiently produce a random matching $M$ such that $\mathbb{E}[\mathbf{1}_M]=(1-\gamma)y$ and for any vector $v\in [0,1]^{A\times B}$, the value $\langle v, \mathbf{1}_M\rangle$ is tightly concentrated around its mean; in particular if $\mu\geq \langle v, y\rangle$ and $t\geq 2$, then
		\[ \mathbb{P}[\langle v,\mathbf{1}_M\rangle \geq t\mu]\leq \exp(-\mu\gamma (2t-3)/20).   \]
	We would like to apply the above concentration bound with $y$ set to be $x'$ and $v$ set to be the indicator of edges in a cut $(S, \overline{S})$; this means $\langle v, \mathbf{1}_M\rangle=|M\cap (S\times \overline{S})|$, the exact quantity we would like to control. However, we need to union bound over all such cuts and there are exponentially many. A common strategy here, also used in the context of the thin tree problem \cite{AGM+17}, is to use the well-known fact due to Karger \cite{Kar93} that the number of near-minimum cuts in any weighted graph is only polynomially large. In particular, if we define $c$ to be the minimum cut
	\[ c\coloneqq \min_{S\neq \emptyset, A\cup B} y(S,\overline{S}), \]
	then the number of sets $S$ where $y(S,\overline{S})\leq \ell\cdot c$ is at most $(2n)^{2\ell}$ for any integer $\ell\geq 1$.
	
	There is a problem with this approach. Unlike in the case of the thin tree problem, there is no a priori lower bound on the minimum cut $c$ of $x'$ or even $x$, which are fractional matchings, and not fractional spanning trees. Indeed, the minimum cut in $x$ and/or $x'$ can be $0$, and the number of cuts of size $0$ can be exponentially large; think of a fractional matching consisting of $\Omega(n)$ disjoint connected components. We overcome this issue by manually removing all small cuts.
	
	Starting with $x'$ we produce a new fractional matching $x''$ as follows: every time there is a cut $(S, \overline{S})$ with value in $(0, \epsilon/10)$, we simply remove all the edges in that cut, i.e., set their value to $0$. We repeat until no such cut can be found and we call the result $x''$. Note that every removal of a cut increases the number of connected components (counted w.r.t.\ the nonzero edges) by one, so the number of iterations can be at most $2n$. In particular, $\sum_e x''_e$ is only smaller than $\sum_e x'_e$ by at most $2n\cdot \epsilon/10=\epsilon n/5$. Thus, $x''$ is a fractional matching of size at least $(1-\epsilon/10)n-\epsilon n/5=(1-3\epsilon/10)n$.
	
	We apply the rounding mechanism of \citet{CVZ11} with $y=x''$ and $\gamma=\epsilon/10$ to get our random matching $M$. Note that the expected size of $M$ is at least $(1-\gamma)(1-3\epsilon/10)n\geq (1-\epsilon/2)n$. Since the size of $M$ can never be larger than $n$, by Markov's inequality, with probability at least $1/2$, the size of $M$ will be $\geq (1-\epsilon)n$. Note that we can boost this probability by repeatedly applying the rounding mechanism of \citet{CVZ11} a number of times; we will show that the probability that $M$ is not thin is negligibly small, so it remains negligibly small even after conditioning on $|M|\geq (1-\epsilon)n$. So repeating the process, at most a logarithmic number of times, until the first time we see a matching of size $\geq (1-\epsilon)n$, gives us the desired result. It just remains to show the probability of $M$ not being thin is negligibly small.
	
	The crucial observation is that it suffices to only consider cuts which nontrivially separate \emph{exactly one} connected component of $x''$. We call such cuts \emph{fundamental cuts}. More precisely, if $C_1,\dots,C_k$ are the connected components of $x''$, then a fundamental cut is a pair $(S, T)$ of disjoint sets such that $S\cup T=C_i$ for some $C_i$. By Karger's arguments, the number of fundamental cuts $(S, T)$ with $x''(S,T)\leq \ell \epsilon/10$ is bounded by $\sum_i |C_i|^{2\ell}\leq (2n)^{2\ell}$. We claim that if for each fundamental cut $(S,T)$ we have, for a logarithmically large $\beta$ to be set later, that
	\[ M(S, T)\leq \beta x(S, T), \]
	then the same inequality holds over \emph{all} global cuts $(U, \overline{U})$, that is
	\[ M(U, \overline{U})\leq \beta x(U, \overline{U}). \]
	Consider such a global cut; for each component $C_i$, the cut $(U, \overline{U})$ induces a fundamental cut $(S_i, T_i)$.\footnote{The cut $(S_i, T_i)$ could be trivial, i.e., one of $S_i$ or $T_i$ could be $\emptyset$, but the arguments still hold.} Note that the edges of the matching $M$ are a subset of the support of $x''$, because $\mathbb{E}[\mathbf{1}_M]=(1-\gamma)x''$. This means that
	\[ M(U, \overline{U})=\sum_i M(S_i, T_i)\leq \sum_i \beta x(S_i, T_i)\leq \beta x(U, \overline{U}), \]
	where the last inequality follows from the fact that any edge with endpoints in $S_i$ and $T_i$ must cross $(U, \overline{U})$.
	
	So it remains to show that the inequality $M(S, T)\leq \beta x(S, T)$ holds for all fundamental cuts with high probability. We derive these by chaining $M(S, T)\leq (\beta/2) x''(S, T)$ with $x''(S, T)\leq x'(S, T)\leq 2 x(S, T)$.
	
	The inequality $x'(S, T)\leq 2 x(S, T)$ holds with high probability if the number of independent trials $h$ of TruncatedRSD that we run to obtain $x'$  is large enough. In particular, the value of the output matching of TruncatedRSD in the fundamental cut $(S, T)$ is a random variable supported on $[0, n]$ with mean $x(S, T)$. So by the standard Hoeffding's inequality, we have
	\[ \mathbb{P}[x'(S, T)-x(S, T)\geq t]\leq \exp(-\Omega(ht^2/n^2)). \]
	For any fundamental cut, we have $x'(S, T)\geq \epsilon/10$. So if we set $t=\epsilon/20$, we get
	\[ \mathbb{P}[x'(S, T)>2x(S, T)]\leq \mathbb{P}[x'(S, T)-x(S, T)> \epsilon/20]\leq \exp(-\Omega(h\epsilon^2/n^2)). \]
	There are at most $2^{O(n)}$ fundamental cuts, so we have
	\[ \mathbb{P}[\exists \text{fundamental cut }(S, T):x'(S, T)>2x(S, T)]\leq \exp(O(n)-\Omega(h\epsilon^2/n^2)).  \]
	It is enough to set $h=\poly(n/\epsilon)$ for the above to be negligibly small: $\exp(-\poly(n))$.
	
	It remains to show that with high probability for all fundamental cuts $(S, T)$ we have $M(S, T)\leq (\beta/2) x''(S, T)$. Here we use the tail bound of \cite{CVZ11} together with the bound on the number of fundamental cuts of various sizes. Consider an integer $\ell\geq 1$. The number of fundamental cuts $(S, T)$ with $x''(S, T)\in [\ell \epsilon/10, (\ell+1)\epsilon/10)$ is at most $n^{O(\ell)}$. For each such cut, we get that
	\[ \mathbb{P}[M(S,T)\geq \beta x''(S, T)/2]\leq \mathbb{P}[M(S,T)\geq \beta \ell \epsilon/20]\leq \exp(-\Omega(\beta\ell \epsilon^2)). \]
	So union bounding over all such cuts, and then summing over all $\ell$, we get
	\[ \mathbb{P}[\exists \text{fundamental cut }(S, T): M(S, T)\geq \beta x''(S, T)/2]\leq \sum_{\ell=1}^\infty n^{O(\ell)}\exp(-\Omega(\beta \ell \epsilon^2)). \]
	For any constant $C>0$, there is some constant $C'$ such that if we set $\beta=C'\log n/\epsilon^2$, we can bound the above as
	\[ \sum_{\ell=1}^\infty \exp(O(\ell \log n)-\Omega(C' \ell \log n))\leq \sum_{\ell=1}^\infty \exp(-C\ell \log n)=O(1/n^C). \qedhere \]
\end{proof}
Note that \Cref{thm:high-prob} also implies the existence of a single near-perfect matching whose distortion is at most $O(\log^2(n))$, a nontrivial fact by itself. As a simple corollary, we get a deterministic, albeit not computationally efficient, mechanism with a similar distortion.
\begin{corollary}
	For any constant $\epsilon>0$, there exists a deterministic exponential time mechanism that outputs a matching of size $(1-\epsilon)n$ guaranteed to have distortion at most $O(\log^2 n)$.
\end{corollary}
\begin{proof}
	We can simply compute the expected value $x$ of the TruncatedRSD mechanism, then enumerate all near-perfect matchings $M$, and for each one check whether $M(S,\overline{S})\leq O(\log^2 n)x(S,\overline{S})$ by enumerating all cuts $S$. We output the first matching $M$ which satisfies this thinness condition. By \Cref{thm:high-prob}, such a near-perfect matching is guaranteed to exist.
\end{proof}

\section{Discussion}\label{sec:discussion}

We discuss some qualitative interpretations of our work, as well as various future directions. For ease of exposition, we split \Cref{q:main} into two open problems separating the question of truthfulness.

\begin{openproblem}
What is the optimal metric distortion for deterministic and randomized mechanisms for min-cost ordinal matching?
\end{openproblem}

\sloppy 
On this problem, we narrow the interval of the optimal distortion for deterministic matchings to $[\Omega(\log n), O(n^2)]$ and for randomized matchings to $[\Omega(\log n), O(n)]$ (using the upper bound proven by \cite{CFRF+16} for RSD). Note that our super-constant lower bound contrasts with several other related metric distortion problems (matching with known item locations, maximum matching, elections) where constant upper bounds are known. Our work in  \Cref{sec:thin} also suggests that it could be quite possible that the optimal distortions for deterministic and randomized mechanisms are the same up to $\log$ factors.

A priori, it is difficult to tell which of the two bounds is closer to the right answer. The lower bound could potentially be improved by considering hard metrics with a greater range of distances (ours only use 0, 1, and 2), but one challenge is that larger distances constrain the preference lists more severely. This makes it difficult to have a variety of very different possible underlying metric spaces that all have large distances in different places. On the other hand, the upper bound could be improved by formalizing this intuition for why lower bounds are difficult into a more clever charging scheme. 

\begin{openproblem}
What is the optimal metric distortion for deterministic and randomized \emph{truthful} mechanisms for min-cost ordinal matching?
\end{openproblem}

We showed that a large class of deterministic mechanisms, which capture a number of known truthful mechanisms, are all as bad as SD in terms of distortion. Could it be that SD is the best deterministic truthful mechanism for this problem? 

It is not hard to see that for a mechanism to be truthful, it must give each agent a set of items that only depends on the other agent's preferences, and then assign the agent their favorite item in the set (the set is exactly all of the items that the agent could be matched to, over all of their possible preference lists). The challenge is that it is difficult to conjure such mechanisms that aren't serializable since the most natural way to prevent an agent from influencing their own set is for decisions to be made in some sequential manner. One can find with a computer search that truthful non-serializable mechanisms do exist, but they are unnatural.

As a counterpoint, part of the issue may just be a lack of imagination, where our attempts to design truthful mechanisms tend to skew toward serialization. Resolving this question may require the development of exciting new kinds of truthful matching mechanisms.

\section*{Acknowledgements} 

We are grateful to the anonymous reviewers for their helpful comments and references. Nima Anari is supported by NSF-CAREER award CCF-2045354 and a Sloan Research Fellowship. Moses Charikar is supported by a Simons Investigator Award. Prasanna Ramakrishnan is supported by Moses Charikar's Simons Investigator Award and Li-Yang Tan's NSF awards 1942123, 2211237, and 2224246. 

\bibliography{pras}
\bibliographystyle{alpha}

\appendix

\section{Lower bound against the Boston mechanism}\label{app:bos}

In this section we briefly discuss the distortion of the Boston mechanism \cite{AS03}, a commonly used mechanism in the school choice literature which seems very natural for this problem but unfortunately has exponential distortion. 

The mechanism proceeds in rounds as follows. In the first round, each agent proposes to their favorite item. Then each item that has some proposers chooses one to irrevocably match to and rejects all others. In the second round, the agents that didn't get matched then propose to their second favorite items, and the mechanism continues in the same way until all agents are matched. Note that the way that the items choose a proposer can be done in a variety of ways: arbitrarily, randomly, according to some fixed or randomized preference ordering, or in the setting with two-sided preferences, according to the preference list of the item. For ease we will say that choices are made according to a fixed priority order, but this can easily be extended to the other settings.

This mechanism seems like a very natural alternative to Serial Dictatorship -- by allowing matchings to happen in parallel, it might cut off the domino effect that causes Serial Dictatorship to have exponential distortion. However, the same effect can still happen, if to a lesser extent. Consider the following instance. The metric space is the line, with the agents are at the points $(1, 2, 4, ..., 2^{k-1})$ and the items lie at the points $(-\epsilon, 2, 4, ..., 2^{k - 1})$. For $t \geq 1$ there are $t$ agents and $t$ items at the point $2^t$. There is one agent at 1, and one item at $-\epsilon$. The co-located agents have identical preferences over the items. Agents have priority in order of their distance to 0.

Then the mechanism will operate as follows. In the first round, for every $t \geq 2$, the agents at $2^t$ all propose to the same item at $2^t$, and that item chooses one of these agents. On the other hand, both the agents at $1$ and $2$ propose to the item at $2$, and the item will choose the agent at $1$. In the second round, for every $t \geq 3$ the agents at $2^t$, the unmatched agents at $2^t$ all propose to the same item at $2^t$, and that item chooses one of these agents. On the other hand, both the unmatched agents at $2$ and $4$ propose to the item at $4$, and the item will choose the agent at $2$. Thus, essentially the same domino effect takes place as with serial dictatorship, only the number of agents and items at each location must scale linearly to prevent the effect from getting cut off. 

It's not hard to see that the overall cost of the mechanism will be $2^{k} + 1$. As the number of agents is $n = k(k-1)/2$, this shows that the mechanism has distortion $2^{\Omega(\sqrt{n})}$. A similar bound can be proven for the randomized variants of the mechanism just by increasing the number of agents and items at each location. This makes the costly matchings likely enough to cause large distortion in expectation.

\end{document}